\documentclass[lettersize,journal]{IEEEtran}
\usepackage{amsmath,amsfonts}
\usepackage{mathrsfs}

\usepackage{algorithm}
\usepackage{algpseudocode}
\usepackage{algorithm}
\usepackage{tabularx}
\usepackage{array}

\usepackage{caption}
\usepackage{float} 
\usepackage{subcaption}

\usepackage{color}
\usepackage{graphicx}
\usepackage{textcomp}
\usepackage{stfloats}
\usepackage{url}
\usepackage{amsthm,amsmath,amssymb}
\usepackage{verbatim}
\usepackage{cite}
\usepackage{blkarray}
\newtheorem{theorem}{Theorem}
\newtheorem{assumption}{Assumption}
\newtheorem{condition}{Condition}

\newtheorem{lemma}[theorem]{Lemma}

\newtheorem*{remark}{Remark}
\newtheorem{definition}{Definition}
\begin{document}

\title{Analysis of Efficient Scheduling in Layered \\ Decoding of GLDPC Codes}

\author{Qingqing Peng, Dongxu Chang, Guiying Yan, Guanghui Wang
\thanks{This work is partially supported by the National Key R\&D Program of China, (2023YFA1009600).

Q. Peng, D. Chang, and G. Wang are with the School of Mathematics, Shandong University, Shandong, China (e-mail:  pqing@mail.sdu.edu.cn; 
dongxuchangmail@163.com;
ghwang@sdu.edu.cn;).

G. Yan is with the Academy of Mathematics and Systems Science, CAS, University of Chinese Academy of Sciences, Beijing, 100190 China (e-mail: yangy@amss.ac.cn).}
}

\markboth{Journal of \LaTeX\ Class Files,~Vol.~14, No.~8, August~2024}%
{Shell \MakeLowercase{\textit{et al.}}: A Sample Article Using IEEEtran.cls for IEEE Journals}


\maketitle

\begin{abstract}
 In this study, we investigate the characteristics of scheduling sequences that enable efficient decoding of generalized low-density parity-check (GLDPC) codes under the layered message-passing algorithm. In particular, we show that scheduling sequences leading to higher decoding efficiency should prioritize the update of constraint nodes corresponding to subcodes with larger minimum distance, fewer minimum-weight codewords, and shorter code length. Based on these characteristics, we design a scheduling algorithm, which further demonstrates the effectiveness of these characteristics through simulation experiments.
\end{abstract}


\section{Introduction}
\IEEEPARstart{A}{s} a variant of the low-density parity-check (LDPC) codes introduced by Gallager \cite{gallager1962low}, generalized low-density parity-check (GLDPC) codes were first proposed by Tanner \cite{tanner1981recursive}. While retaining the sparse graph representation, GLDPC codes replace single parity-check (SPC) nodes with error-correcting block subcodes, known as generalized constraint (GC) nodes. The degrees of GC nodes match the length of their associated subcodes. GLDPC codes offer advantages such as the ability to employ more powerful decoders at GC nodes during decoding, resulting in improved performance \cite{liu2019probabilistic,chang2025gap}, faster convergence \cite{mulholland2015design}, and reduced error floor \cite{liva2008quasi,mitchell2013minimum}.

However, for GLDPC codes, the use of powerful decoders at GC nodes, such as a posteriori probability (APP) decoders, entails much higher computational complexity, thereby motivating the need for improved decoding efficiency. An effective approach for improving the efficiency of message-passing decoders is the use of the layered strategy, first proposed by D. Hocevar \cite{hocevar2004reduced}. Unlike the flooding strategy adopted by conventional message-passing decoders, which updates all variable-to-check (V2C) or check-to-variable (C2V) messages simultaneously in a single iteration, layered message-passing decoding updates messages sequentially. This sequential update can reduce decoding complexity by nearly $50\%$ while achieving comparable performance.

Despite its benefits, the layered message-passing algorithm does not have a specific design for the order of the decoding sequence. Many studies have focused on optimizing the order of decoding sequences to further reduce the decoding complexity of conventional LDPC codes. For instance, previous studies have demonstrated that prioritizing the update of constraint nodes with lower degree \cite{frenzel2019static}, or those connected to the least-punctured edges and to variable nodes with the highest-degree (LPHD) \cite{wang2020two}, can lead to more efficient decoding schedules. In \cite{chang2025analysis}, it is suggested that a ``good" scheduling sequence should prioritize updating constraint nodes with lower error probabilities aggregated from neighboring variable nodes. Building on this principle, the authors further demonstrate that certain observable features of ``good'' scheduling sequences—such as row weights, rows connected to punctured columns, and column weights can serve as corollaries of this characteristic.

Although many scheduling algorithms have been proposed to reduce decoding complexity, most of them are developed based on conventional LDPC codes \cite{frenzel2019static,wang2020two,chang2025analysis,tian2022novel,chang2024analysis}, with limited investigations conducted for GLDPC codes. In this paper, we investigate scheduling strategies for GLDPC codes under layered message-passing decoding. In contrast to conventional LDPC codes, the presence of GC nodes introduces additional structural factors, such as minimum distance and weight distribution, which influence decoding performance in a non-trivial manner. Over the binary-input additive white Gaussian noise (BI-AWGN) channel and the binary
erasure channel (BEC), we show that effective scheduling should prioritize constraint nodes with larger minimum distance, fewer minimum-weight codewords, and shorter code length, with minimum distance being the dominant factor. Based on this conclusion, we propose a scheduling algorithm that prioritizes constraint nodes based on the ordering of their associated subcodes with respect to minimum distance, the number of minimum-weight codewords, and code length, with minimum distance serving as the dominant consideration. Simulation results further verify the effectiveness of the proposed strategy.

\section{Preliminaries}
\subsection{LDPC code and GLDPC Code Ensembles}

In \cite{tanner1981recursive}, Tanner generalized LDPC codes by introducing the use of block codes as constraint nodes, which are referred to as GC nodes. Each GC node is associated with a subcode, and the variable nodes connected to the GC node are required to satisfy the parity-check constraints of that subcode, rather than the simple parity constraint of a SPC node.


Let {$G_{R_M}$} be an LDPC code with $M$ constraint nodes, and let $\mathcal{C}^{(1)}, \ldots, \mathcal{C}^{(M)}$ denote a family of linear codes (subcodes) such that the length of $\mathcal{C}^{(i)}$ equals the degree of the $i$th constraint node in {$G_{R_M}$}. The $({G_{R_M}}, \mathcal{C}^{(1)}, \ldots, \mathcal{C}^{(M)})$ GLDPC ensemble is defined as follows:

\begin{definition}[GLDPC code ensemble]
In the $i$th row of the parity-check matrix of ${G_{R_M}}$, each entry equal to 1 is replaced by a distinct column of the parity-check matrix of $\mathcal{C}^{(i)}$, with each column used exactly once, while each entry equal to 0 is replaced by a zero column vector. By randomly assigning the columns of the parity-check matrix of $\mathcal{C}^{(i)}$ to the positions corresponding to the variable nodes, one obtains an element of the $({G_{R_M}}, \mathcal{C}^{(1)}, \ldots, \mathcal{C}^{(M)})$ GLDPC ensemble.
    \label{def_ensemble}
\end{definition}

Note that although GC node constraints are represented in the parity-check matrix by replacing them with the corresponding subcode parity-check matrices, each GC node is treated as a single entity during decoding. $\mathcal{C}^{(i)}$ may be chosen as an SPC code. When $\mathcal{C}^{(1)}, \ldots, \mathcal{C}^{(M)}$ are all SPC codes, the resulting random GLDPC code reduces to the LDPC code ${G_{R_M}}$.

\subsection{Flooding Message-Passing Decoder}
Assume that the codeword $\boldsymbol{c}=\left\{c_1,\cdots,c_N\right\}\in \left\{0,1\right\}^N$ is transmitted through a binary memoryless symmetric (BMS) channel, and $\boldsymbol{y}=\left\{y_1,\cdots,y_N\right\}$ is the received sequence. Let
\[
		L_\alpha^{(0)}=\ln{\frac{P(y_\alpha|c_\alpha=0)}{P(y_\alpha|c_\alpha=1)}}
\]
denote the LLR associated with the $\alpha$-th variable node.

In LDPC codes, the flooding message-passing algorithm proceeds through iterative message exchanges between variable nodes and check nodes, referred to as V2C and C2V message updates, respectively. The update rule for the V2C messages is
\begin{equation}
		L_{\alpha \rightarrow \eta}^{(l)}=L_{\alpha}^{(0)}+\sum_{g \in \mathcal{N}_\alpha \backslash \eta} L_{g \rightarrow \alpha}^{(l-1)},
  \label{V2C message update_awgn}
	\end{equation}
The message update rule of the C2V message is 
 \begin{equation}
		L_{\eta \rightarrow \alpha}^{(l)}=2 \tanh ^{-1}\left(\prod_{h \in \mathcal{N}_\eta \backslash \alpha} \tanh \left(L_{h \rightarrow \eta}^{(l)} / 2\right)\right)
  \label{C2V message update_awgn}
	\end{equation}
where $l$ is the number of iterations, $\mathcal{N}_v$ represents the nodes connected directly to node $v$, $\alpha \rightarrow \eta$ means from variable node $\alpha$ to constraint node $\eta$ and $\eta \rightarrow \alpha$ means from constraint node $\eta$ to variable node $\alpha$. The initial message are given by $L_{\alpha \rightarrow \eta}^{(0)}=0$ and $L_{\eta \rightarrow \alpha}^{(0)}=0$.

 
GLDPC codes are also decoded using a message-passing algorithm consisting of V2C and C2V message updates. The V2C update rule is identical to that of LDPC codes and can be computed according to (\ref{V2C message update_awgn}). However, the constraint nodes in GLDPC codes include both SPC nodes and GC nodes, and the C2V message update rules depend on the type of constraint node. For SPC nodes, the outgoing messages are determined based on the received messages and the corresponding SPC code constraints, and the C2V update rule is the same as that in (\ref{C2V message update_awgn}). In contrast, for the GC nodes, the transmission of messages is managed by their corresponding subcodes and decoded using the a posteriori probability (APP) decoder. The decoding rules of the APP decoder are as follows \cite{richardson2008modern}: For a GC node $C$ with subcode $\mathcal{C}$, let $v_i$ denotes its $i$th neighboring variable node, $\boldsymbol{L}_{\sim i}^{(l)}$ denotes the message received from its neighboring variable nodes other than $v_i$ at iteration $l$, and $\boldsymbol{c}'$ denotes the codewords of $\mathcal{C}$. The message to $v_i$ from $C$ under the APP decoder is
\begin{equation}
\begin{split}
    L_{C\rightarrow v_i}^{(l)}& = \ln \frac{P(\boldsymbol{L}_{\thicksim i}^{(l)}|v_i=0)}{P(\boldsymbol{L}_{\thicksim i}^{(l)}|v_i=1)}\\
    &=\ln \frac{\sum\limits_{\boldsymbol{c}':c'_i=0}\exp(\sum\limits_{j\neq i} (1-c'_j)L_{v_j \rightarrow C}^{(l)}) }{\sum\limits_{\boldsymbol{c}':c'_i=1}\exp (\sum\limits_{j\neq i} (1-c'_j)L_{v_j \rightarrow C}^{(l)})}\\
    &\approx \min_{{\boldsymbol{c}':c'_i=1}}\sum_{j \neq i}c'_jL_{v_j \rightarrow C}^{(l)}-\min_{{\boldsymbol{c}':c'_i=0}}\sum_{j \neq i}c'_jL_{v_j \rightarrow C}^{(l)}
\end{split}
\label{APP_awgn}
\end{equation}

Note that an SPC node and its neighboring variable nodes form a subcode with only a single parity check, which therefore has a minimum distance of 2. In this case, update rules (\ref{C2V message update_awgn}) and (\ref{APP_awgn}) coincide.

\subsection{Layered Message-passing Decoder and Scheduling Sequences}
In each decoding iteration, the layered message-passing decoder \cite{hocevar2004reduced} sequentially selects constraint nodes in the graph and updates the information on all edges connected to the selected constraint node. Specifically, it performs V2C message updates along all edges connected to the selected constraint node, followed by C2V message updates. The order in which constraint nodes are selected for updating in the layered message-passing decoder is referred to as the scheduling sequence, defined as follows:

\begin{definition}[scheduling sequences]
	The scheduling sequence is a sequence composed of the indices of constraint nodes. It represents the order in which constraint nodes are selected for decoding in the layered message-passing decoder, with the decoder selecting constraint nodes for decoding according to the order in which they appear in the scheduling sequences.
\end{definition}
In the subsequent analysis, assuming that the same message-passing order is applied in each iteration, the scheduling sequence can be defined solely by the update order of the constraint nodes in the first iteration.

\section{Scheduling Characteristics under BI-AWGN channel}

Given different scheduling sequences, the same number of message-passing iterations may lead to different performance. As a result, different constraint nodes should be scheduled at distinct positions to maximize their contribution to the message-passing process. We claim that the optimal scheduling position of a constraint node is largely influenced by the minimum distance of its associated subcode, the number of its minimum-weight codewords, and its code length.

To ensure clarity in the subsequent description and proof, for any GLDPC code $\mathcal{G}$, we assume that decoding takes place at any two adjacent constraint nodes $a$ and $b$ of $\mathcal{G}$ in the scheduling sequence. For the subcodes associated with these two constraint nodes, we use superscripts or subscripts $a$ and $b$ to distinguish their respective parameters (e.g., subcode $\mathcal{C}$, minimum distance $d_{\min}$, code length $n$, etc.). Let $v_i^a$ and $v_j^b$ denote the $i$-th and $j$-th variable nodes connected to constraint nodes $a$ and $b$, respectively. Without loss of generality, assume that $v_i^a = v_i^b$ for $i \in \{1, \ldots, n_{ab}\}$, indicating that these variable nodes are simultaneously connected to both constraint nodes $a$ and $b$.

Before investigating the impact of the update order between two constraint nodes that are adjacent in the scheduling sequence on the decoding performance, it is necessary to first analyze the effect of updating a single constraint node. To simplify the analysis, the subsequent discussion will be conducted under the conditions of Assumption \ref{assump_1_awgn}.

\begin{assumption}
On the BI-AWGN channel, assume that the
messages are passed during the decoding process of GLDPC codes
follow Gaussian distributions. In particular, the message \( L_{v_i^a \rightarrow a} \) is modeled as a Gaussian random variable with mean \( u_i^a \) and variance \(2u_i^a \) \cite{chung2001analysis}. The corresponding error probability of this message is given by \(P_{v_i^a \rightarrow a} = Q\left(\sqrt{\frac{u_i^a}{2}}\right).\) Messages from different variable nodes are assumed to be independent and non-zero, with error probabilities tending to zero. Assume that the all-zero codeword is transmitted.
\label{assump_1_awgn}
\end{assumption}

\begin{lemma}
For GLDPC codes that satisfy Assumption \ref{assump_1_awgn}, the erasure probability of the message passed from $a$ to $v_i^a$, denoted by $P_{a \rightarrow v_i^a}$, is $|R_a(\boldsymbol{0})|Q(r_a(\boldsymbol{0}))+o(Q(r_a(\boldsymbol{0})))$, where $r_a(\boldsymbol{z})$ and $R_a(\boldsymbol{z})$ are defined as follows,
\begin{equation}
    r_a(\boldsymbol{z}) = \min_{\boldsymbol{c}\in \mathcal{C}_a:c_i=1} \frac{\sum_{j \neq i}(c_j-z_j)u_j^a}{\sqrt{\sum_{j \neq i}(c_j-z_j)^22u_j^a}}
\end{equation}
\begin{equation}
    R_a(\boldsymbol{z}) = \bigg\{\boldsymbol{c }\in \mathcal{C}_a:c_i=1,\frac{\sum_{j \neq i}(c_j-z_j)u_j^a}{\sqrt{\sum_{j \neq i}(c_j-z_j)^22u_j^a}} =r^a(\boldsymbol{z})\bigg\}
\end{equation}
\label{lemma1_awgn}
\end{lemma}
\begin{proof}
 According to \eqref{APP_awgn}, the condition \( L_{a \rightarrow v_i^a} < 0 \) is equivalent to the following: for any codeword \( \boldsymbol{x} \in \mathcal{C}_a \) such that \( x_i = 0 \), there exists a codeword \( \boldsymbol{c} \in \mathcal{C}_a \) with \( c_i = 1 \) satisfying \(
\sum_{j \neq i} c_j L_{v_j^a \rightarrow a } < \sum_{j \neq i} x_j L_{v_j^a \rightarrow a }. 
\) Define $S_i(\boldsymbol{c},\boldsymbol{x})=\sum_{j\neq i}(c_j-x_j)L_{v_j^a \rightarrow a }.$ Then, the above condition can be equivalently expressed as
\begin{equation}
   P_{a \rightarrow v_i^a}=P\left( \mathop{\cap}_{\boldsymbol{x} \in \mathcal{C}_a:x_i=0} \mathop{\cup}_{\boldsymbol{c} \in \mathcal{C}_a:c_i=1} S_i(\boldsymbol{c},\boldsymbol{x})<0\right)
\end{equation}
Since \( L_{v_j^a \rightarrow a} \sim \mathcal{N}(u_j^a, 2u_j^a) \), it follows that
\(
\sum_{j \neq i} (c_j - x_j) L_{v_j^a \rightarrow a}
\sim \mathcal{N}\left( \sum_{j \neq i} (c_j - x_j) u_j^a, \sum_{j \neq i} (c_j - x_j)^2 \cdot 2u_j^a \right).
\) Therefore, for any such $\boldsymbol{c}$ and $\boldsymbol{x}$ , we have
\[
    P\Big( S_i(\boldsymbol{c},\boldsymbol{x}) < 0 \Big)= Q\bigg( \frac{\sum_{j \neq i} (c_j - x_j) u_j^a}{\sqrt{\sum_{j \neq i} (c_j - x_j)^2 \cdot 2u_j^a}}\bigg)
\]
For any such $\boldsymbol{x}$, we have
\[
    \begin{split}
        &P\left( \mathop{\cup}_{\boldsymbol{c} \in \mathcal{C}_a:c_i=1} S_i(\boldsymbol{c},\boldsymbol{x}) < 0\right)=\sum_{\boldsymbol{c} \in \mathcal{C}_a:c_i=1} P\left( S_i(\boldsymbol{c},\boldsymbol{x}) < 0\right)\\
        &-\sum_{\boldsymbol{c},\boldsymbol{c}'\in \mathcal{C}_a:c_i=1,c_i'=1} P\left(S_i(\boldsymbol{c},\boldsymbol{x}) < 0,S_i(\boldsymbol{c}',\boldsymbol{x}) < 0\right)+\ldots\\
        &=|R_a(\boldsymbol{x})|Q(r_a(\boldsymbol{x}))+o(Q(r_a(\boldsymbol{x}))).
    \end{split}
\]
By the chain rule, we have
\begin{equation}
\begin{split}
   &P_{a \rightarrow v_i^a} = P\left( \mathop{\cup}_{\boldsymbol{c} \in \mathcal{C}_a:c_i=1} S_i(\boldsymbol{c},\boldsymbol{0})<0\right) \\
   &\times P\left(\mathop{\cup}_{\boldsymbol{c} \in \mathcal{C}_a:c_i=1} S_i(\boldsymbol{c},\boldsymbol{x}^{(1)})<0\bigg| \mathop{\cup}_{\boldsymbol{c} \in \mathcal{C}_a:c_i=1} S_i(\boldsymbol{c},\boldsymbol{0})<0\right) \times \cdots
\end{split}
\label{eq_8_awgn}
\end{equation}
where $\boldsymbol{0}$, $\boldsymbol{x}^{(1)}$, $\cdots,$ are codewords in $\mathcal{C}_a$ whose $i$-th components are equal to zero. on the one hand, Since \[P\left(\mathop{\cup}_{\boldsymbol{c} \in \mathcal{C}_a:c_i=1} S_i(\boldsymbol{c},\boldsymbol{x}^{(1)})<0\bigg| \mathop{\cup}_{\boldsymbol{c} \in \mathcal{C}_a:c_i=1} S_i(\boldsymbol{c},\boldsymbol{0})<0\right) < 1,\] it follows that
\begin{equation}
\begin{split}
     &P_{a \rightarrow v_i^a} \le P\left( \mathop{\cup}_{\boldsymbol{c} \in \mathcal{C}_a:c_i=1} S_i(\boldsymbol{c},\boldsymbol{0})<0\right) \\
     &= |R_a(\boldsymbol{0})|Q(r_a(\boldsymbol{0}))+o(Q(r_a(\boldsymbol{0}))).
\end{split}  
\end{equation}
on the other hand, from \eqref{eq_9_awgn}, we obtain the lower bound
\[\begin{split}
    &P_{a \rightarrow v_i^a} \ge P\left( \mathop{\cup}_{\boldsymbol{c} \in \mathcal{C}_a:c_i=1} S_i(\boldsymbol{c},\boldsymbol{0})<0\right)  \\
&\times \bigg( 1-Q\big(\frac{\sum_{j\neq i}x_j^{(1)}u_j^a}{\sqrt{\sum_{j \neq i}|x_j^{(1)}|2u_j^a}}\big)\bigg)\times \cdots \\
&\ge |R_a(\boldsymbol{0})|Q(r_a(\boldsymbol{0}))+o(Q(r_a(\boldsymbol{0}))).
\end{split}\] 
Combining the above results, we obtain
\[
P_{a \rightarrow v_i^a} = |R_a(\boldsymbol{0})|Q(r_a(\boldsymbol{0}))+o(Q(r_a(\boldsymbol{0}))).
\]

\begin{figure*}
\begin{small}
\begin{equation}
\begin{split}
    &  P\left(\mathop{\cup}_{\boldsymbol{c} \in \mathcal{C}_a:c_i=1} S_i(\boldsymbol{c},\boldsymbol{x}^{(1)})<0\bigg| \mathop{\cup}_{\boldsymbol{c} \in \mathcal{C}_a:c_i=1} S_i(\boldsymbol{c},\boldsymbol{0})<0\right) \\
    &\ge P\left(\sum_{j\neq i}x_j^{(1)}L_{v_j^a \rightarrow a}>0\right)P\left(\mathop{\cup}_{\boldsymbol{c} \in \mathcal{C}_a:c_i=1} S_i(\boldsymbol{c},\boldsymbol{x}^{(1)})<0\bigg| \mathop{\cup}_{\boldsymbol{c} \in \mathcal{C}_a:c_i=1} S_i(\boldsymbol{c},\boldsymbol{0})<0,\sum_{j\neq i}x_j^{(1)}L_{v_j^a \rightarrow a}>0\right)\\
    & +  P\left(\sum_{j\neq i}x_j^{(1)}L_{v_j^a \rightarrow a}\le 0\right)P\left(\mathop{\cup}_{\boldsymbol{c} \in \mathcal{C}_a:c_i=1} S_i(\boldsymbol{c},\boldsymbol{x}^{(1)})<0\bigg| \mathop{\cup}_{\boldsymbol{c} \in \mathcal{C}_a:c_i=1} S_i(\boldsymbol{c},\boldsymbol{0})<0,\sum_{j\neq i}x_j^{(1)}L_{v_j^a \rightarrow a}\le0\right) \\
    &\ge P\left(\sum_{j\neq i}x_j^{(1)}L_{v_j^a \rightarrow a}>0\right)P\left(\mathop{\cup}_{\boldsymbol{c} \in \mathcal{C}_a:c_i=1} S_i(\boldsymbol{c},\boldsymbol{x}^{(1)})<0\bigg| \mathop{\cup}_{\boldsymbol{c} \in \mathcal{C}_a:c_i=1} S_i(\boldsymbol{c},\boldsymbol{0})<0,\sum_{j\neq i}x_j^{(1)}L_{v_j^a \rightarrow a}>0\right)\\
    &=P\left(\sum_{j\neq i}x_j^{(1)}L_{v_j^a\rightarrow a}>0\right)=1-Q\bigg(\frac{\sum_{j\neq i}x_j^{(1)}u_j^a}{\sqrt{\sum_{j \neq i}|x_j^{(1)}|2u_j^a}}\bigg)
\end{split} 
\label{eq_9_awgn}
\end{equation}
\end{small}
\end{figure*}
\end{proof}

Next, we consider the impact of the update order between two constraint nodes that are adjacent in the scheduling sequence on the decoding performance. For clarity and notational simplicity, we focus on the first decoding iteration, in which the messages from all variable nodes are assumed to be independent and identically distributed Gaussian random variables following $\mathcal{N}(u, 2u)$.

\begin{theorem}
\label{theo1}
For GLDPC codes that satisfy Assumption \ref{assump_1_awgn}, consider the following two conditions: 

\textbf{Condition 1:} $d_{\min}^a< d_{\min}^b$;

\textbf{Condition 2:} $d_{\min}^a= d_{\min}^b$ and  $\sum_{i={n_{ab}+1}}^{n_a}(A_i^a-B_i^a) > \sum_{k={n_{ab}+1}}^{n_b}(A_k^b-B_k^b)$;

\noindent where \( A_i^a \) denotes the number of codewords in $\mathcal{C}_a$ with weight $d_{\min}^a$ and $v_i^a =1$. The quantity \( B_i^a \) denotes the number of such codewords whose support is contained in $\{n_{ab}+1,\cdots,n_a\}\cup \{i\}$. 

If either of the above conditions holds, updating constraint node $b$ before $a$ results in a lower average error probability across the variable nodes, compared to updating $a$ before $b$.
\end{theorem}
\begin{proof}
We first analyze the decoding performance when the constraint node $a$ is updated before $b$. For the constraint node $a$, according to Lemma \ref{lemma1_awgn}, the error probability of the message passed from $a$ to $v_i^a$ is given by:
\begin{equation}
  \begin{split}
       & P_{a\rightarrow v_{i}^a}=|R_a(\boldsymbol{0})|Q(r_a(\boldsymbol{0}))+o(Q(r_a(\boldsymbol{0})))\\&=Q\Big(\sqrt{d_{\min}^au/2}\Big)A_i^a +o\Big(Q\Big(\sqrt{d_{\min}^au/2}\Big)\Big).
  \end{split}
\end{equation}
where
\[
\begin{split}
    r_a(\boldsymbol{0}) &= \min_{\boldsymbol{c}\in \mathcal{C}_a:c_i=1} \frac{\sum_{j \neq i}c_ju}{\sqrt{\sum_{j \neq i}c_j2u}} \\
    &=\min_{\boldsymbol{c}\in \mathcal{C}_a:c_i=1} \sqrt{\frac{\sum_{j \neq i}c_ju}{2}} \\
    &=\min_{\boldsymbol{c}\in \mathcal{C}_a:c_i=1} \sqrt{\frac{w(\boldsymbol{c})u}{2}}
\end{split}
\]
and $w(\boldsymbol{c})$ denotes the Hamming weight of $\boldsymbol{c}$.

Then, after updating constraint node $a$, the messages \( L_{v_j^b \rightarrow b} \) follow Gaussian distributions \( \mathcal{N}(u_j^b>2u, 2u_j^b) \) for all \( j \in \{1, \ldots, n_{ab}\} \). Similarly, according to Lemma \ref{lemma1_awgn}, the error probability of the message passed from $b$ to $v_k^b$ is given by:
\begin{equation}
    P_{b\rightarrow v_{k}^b}=Q\Big(\sqrt{d_{\min}^bu/2}\Big)B_k^b +o\Big(Q\Big(\sqrt{d_{\min}^bu/2}\Big)\Big),
\end{equation}
Therefore, after sequentially updating constraint nodes $a$ and $b$, the sum of the error probabilities over all variable nodes connected to $a$ or $b$ is given by:

\begin{equation}
\begin{split}
    &P_{\text{sum}}^{ab} = \sum_{j=1}^{n_{ab}}Q\Big(\sqrt{u/2}\Big)P_{a\rightarrow v_{j}^a}P_{b\rightarrow v_{j}^b} \\
    &+\sum_{i=n_{ab}+1}^{n_a} Q\Big(\sqrt{u/2}\Big)P_{a\rightarrow v_{i}^a}+\sum_{k={n_{ab}+1}}^{n_b}Q\Big(\sqrt{u/2}\Big)P_{b\rightarrow v_{k}^b} \\
    &=Q\Big(\sqrt{u/2}\Big)Q\Big(\sqrt{d_{\min}^a u/2}\Big)\sum_{i={n_{ab}+1}}^{n_a}A_i^a\\
&+Q\Big(\sqrt{u/2}\Big)Q\Big(\sqrt{d_{\min}^b u/2}\Big)\sum_{k={n_{ab}+1}}^{n_b}B_k^b \\
&+o\bigg(Q\Big(\sqrt{u/2}\Big)Q\Big(\sqrt{\min \{d_{\min}^a,d_{\min}^b\} u/2}\Big)\bigg)
\end{split}
\end{equation}
Similarly, the sum of the error probabilities after sequentially updating constraint nodes \(b\) and \(a\), denoted by \(P_{\text{sum}}^{ba}\), can be obtained:

\begin{equation}
\begin{split}
    &P_{\text{sum}}^{ba}=Q\Big(\sqrt{u/2}\Big)Q\Big(\sqrt{d_{\min}^a u/2}\Big)\sum_{i={n_{ab}+1}}^{n_a}B_i^a\\
&+Q\Big(\sqrt{u/2}\Big)Q\Big(\sqrt{d_{\min}^b u/2}\Big)\sum_{k={n_{ab}+1}}^{n_b}A_k^b \\
&+o\bigg(Q\Big(\sqrt{u/2}\Big)Q\Big(\sqrt{\min \{d_{\min}^a,d_{\min}^b\} u/2}\Big)\bigg)
\end{split}
\end{equation}
When condition 1 or 2 holds, we have \( P_{\text{sum}}^{ab} \ge P_{\text{sum}}^{ba} \), which completes the proof of the theorem.
\end{proof}


Theorem 2 implies that, for any pair of adjacent subcodes, the one with better distance properties should be updated first. For non-adjacent constraint nodes, this ordering can be achieved through successive exchanges of adjacent nodes.


\begin{remark}
Under the GLDPC code ensemble, the comparison of
\(\sum_{i=n_{ab}+1}^{n_a} (A_i^a - B_i^a)\) in Condition 2 can be interpreted as comparing the quantity \(
f(A_{\min}^a,d_{\min},n_a,n_{ab}) = (n_a - n_{ab}) \frac{A_{\min}^a}{{n_a \choose d_{\min}^a}}\big({n_a -1 \choose d_{\min}^a-1}-{n_a-n_{ab}-1 \choose d_{\min}^a-1}\big)\), where \(A_{\min}^a\) denotes the number of minimum-weight codewords in \(\mathcal{C}_a\). Here, \(\frac{A_{\min}^a}{{n_a \choose d_{\min}^a}}\) can be viewed as the probability that a randomly chosen weight-\(d\) binary vector is a codeword in \(\mathcal{C}_a\). Moreover, \(
\frac{A_{\min}^a}{{n_a \choose d_{\min}^a}} {n_a - 1 \choose d_{\min}^a - 1}
\) represents the average number of minimum-weight codewords whose $i$-th bit is equal to 1.
\end{remark}
\begin{remark}
By viewing an LDPC code as a GLDPC code in which each constraint node corresponds to an SPC code, prioritizing lower-degree constraint nodes in conventional LDPC decoding \cite{frenzel2019static} follows as a corollary of Theorem 2.
\end{remark}

\section{Scheduling Characteristics under BEC}
\subsection{Flooding Message-Passing Decoder under BEC}
Assume that the codeword $\boldsymbol{c}=\left\{c_1,\cdots,c_N\right\}\in \left\{0,1\right\}^N$ is transmitted through a BEC, and $\boldsymbol{y}=\left\{y_1,\cdots,y_N\right\}\in \left\{0,1,?\right\}^N$ is the received sequence, where $?$ indicates an erasure.

In LDPC codes, the flooding message-passing algorithm proceeds through iterative message exchanges between variable nodes and check nodes, referred to as V2C and C2V message updates, respectively. The update rule for the V2C messages is

\begin{equation}
m_{\alpha \rightarrow \eta}^{(l)} =
\begin{cases}
?, &
  \text{if } \forall h \in \mathcal{N}(\alpha)\setminus \{\eta\},m_{h \rightarrow \alpha}^{(l-1)} = ?,\text{ and } y_\alpha = ?, \\[2ex]
\beta, &
  \text{if } \exists \, h \in \mathcal{N}(\alpha)\setminus \{\eta\} 
  \text{ such that } m_{h \rightarrow \alpha}^{(l-1)} = \beta,\\
  &\text{or } y_\alpha = \beta \text{ with } \beta \in \{0,1\}.
\end{cases}
\label{V2C message update}
\end{equation}
The message update rule of the C2V message is 
\begin{equation}
		m_{\eta \rightarrow \alpha}^{(l)}= \begin{cases}
		    ? ,& \text{if }\exists \, g \in \mathcal{N}(\eta) \backslash \alpha,\\
            &\text{such that }m_{g \rightarrow \eta}^{(l)}=?   \\
            \oplus_{g \in \mathcal{N}(\eta) \backslash \alpha}m_{g \rightarrow \eta}^{(l)},& \text{otherwise.}
		\end{cases}
  \label{C2V message update}
\end{equation}
where $l$ is the number of iterations, $\mathcal{N}(v)$ represents the nodes connected directly to node $v$, $\alpha \rightarrow \eta$ means from variable node $\alpha$ to constraint node $\eta$ and $\eta \rightarrow \alpha$ means from constraint node $\eta$ to variable node $\alpha$, $\oplus$ denotes modulo-2 addition. The initial message are given by $m_{\alpha \rightarrow \eta}^{(0)}=y_\alpha$ and $m_{\eta \rightarrow \alpha }^{(0)}=?$.
 
GLDPC codes are also decoded using a message-passing algorithm consisting of V2C and C2V message updates. The V2C update rule is identical to that of LDPC codes and can be computed according to (\ref{V2C message update}). However, the constraint nodes in GLDPC codes include both SPC nodes and GC nodes, and the C2V message update rules depend on the type of constraint node. For SPC nodes, the outgoing messages are determined based on the received messages and the corresponding SPC code constraints, and the C2V update rule is the same as that in (\ref{C2V message update}). In contrast, for the constraint nodes, the transmission of messages is managed by their corresponding subcodes and decoded using the APP decoder. The decoding rules of the APP decoder are as follows \cite{richardson2008modern}: For a constraint node $C$ with subcode $\mathcal{C}$, let $v_i$ denote its $i$th neighboring variable node, $\boldsymbol{m}_{\sim i}$ denote the message received from its neighboring variable nodes other than $v_i$, and $\boldsymbol{x}$ denote the codewords of $\mathcal{C}$. Define
$$
\text{Er}(\boldsymbol{m}_{\sim i}) \triangleq \{ j \mid j\neq i \text{, and }m_j \ \text{is erased} \},
$$
$$
\text{Ne}(\boldsymbol{m}_{\sim i}) \triangleq \{ j \mid j\neq i \text{, and }m_j \ \text{is not erased} \},
$$
$$
\chi(\boldsymbol{m_{\sim i}})\triangleq\{\boldsymbol{x} \in \mathcal{C}:\boldsymbol{x}_{\text{Ne}(\boldsymbol{m}_{\sim i})}=\boldsymbol{m}_{\text{Ne}(\boldsymbol{m}_{\sim i})}\},
$$
where $\text{Er}(\boldsymbol{m}_{\sim i})$ denotes the set of indices (excluding $i$) at which the components of $\boldsymbol{m}_{\sim i}$ are erased, while $\text{Ne}(\boldsymbol{m}_{\sim i})$ denotes the complementary set of indices corresponding to non-erased positions.

Then, the message to $v_i$ from $C$ under the APP decoder is 
\begin{equation}
 m_{C \rightarrow v_i} = \begin{cases}
     \beta,&\text{if } \forall \boldsymbol{x} \in \chi(\boldsymbol{m}_{\sim i}),x_i=\beta\\
     ?, &\text{otherwise.}
 \end{cases}
\label{APP}
\end{equation}

Note that an SPC node and its neighboring variable nodes form a subcode with only a single parity check, which therefore has a minimum distance of 2. In this case, update rules (\ref{C2V message update}) and (\ref{APP}) coincide.

\subsection{Scheduling Characteristics}

Before investigating the impact of the update order between two constraint nodes that are adjacent in the scheduling sequence on the decoding performance, it is necessary to first analyze the effect of updating a single constraint node. To simplify the analysis, the subsequent discussion will be conducted under the conditions of Assumption \ref{assump_3}.

\begin{assumption}
In the first iteration, assume that the messages passed from different variable nodes to nodes $a$ and $b$ are independent of each other. Let $\varepsilon$ denote a vanishing parameter characterizing the error probability of the incoming messages. Then, $\varepsilon_i^a = \Theta(\varepsilon)$ for all $i$, and $\varepsilon_j^b = \Theta(\varepsilon)$ for all $j$ where $\varepsilon \to 0$.
\label{assump_3}
\end{assumption}

\begin{condition}
Let $\boldsymbol{m}_{\thicksim i}^a$ denote the message transmitted from the neighboring variable nodes of $a$, excluding $v_i^a$ to $a$. We say that message $\boldsymbol{m}_{\thicksim i}^a$ satisfies Condition 1 if there exists a codeword $\boldsymbol{u}^a \in \mathcal{C}_{a}$ with $u_i^a=1$ such that the support set of $\boldsymbol{u}_{\sim i}^a$ satisfies supp$(\boldsymbol{u}_{\sim i}^a) \subseteq$ Er$(\boldsymbol{m}^a_{\sim i}) $.
\end{condition}
\begin{lemma}
\label{lemma_biterror_allp}
For GLDPC codes that satisfy Assumption \ref{assump_3}, the erasure probability of the message passed from $a$ to $v_i^a$, denoted by $P_{a \rightarrow v_i^a}$, is $\Theta(\varepsilon^{d_{\min}^a-1})$.
\end{lemma}

\begin{proof}
Assuming that the all-zero codeword is transmitted. According to \eqref{APP}, the decoding result \( m_{a \rightarrow v_i^a} = ? \) if and only if \( \boldsymbol{m}_{\thicksim i}^a \) satisfies Condition 1. Based on Assumption 1, we obtain that
 \begin{equation}
     \begin{split}
         P(\boldsymbol{m}_{\thicksim i}^a|\boldsymbol{0}^{n_a-1})=\prod_{k \neq i}P(m_{v_k \rightarrow a}|0)=\prod_{k \neq i} (\varepsilon_k^a)^{e_k}(1-\varepsilon_k^a)^{1-e_k}
     \end{split}
     \label{e1}
 \end{equation}
where $e_k=1$ if $m_{v_k \rightarrow a}=?$, and $e_k = 0$ otherwise.

First, we establish that $P_{a \rightarrow v_i^a} = O(\varepsilon^{d_{\min}^a-1})$ by proving that there does not exist any message $\boldsymbol{m}_{\sim i}^a$ with fewer than $d_{\min}^a-1$ erasures such that $m_{a \rightarrow v_i^a} = ?$. This is because, for any message \( \boldsymbol{m}_{\sim i}^a \) with fewer than \( d_{\min}^a - 1 \) erasures, and for any nonzero codeword \( \boldsymbol{x}^a \in \mathcal{C}_a \), the cardinality of the support set satisfies \( |\mathrm{supp}(\boldsymbol{x}_{\sim i}^a)| \ge d_{\min}^a - 1 \). Meanwhile, since \( |\mathrm{Er}(\boldsymbol{m}_{\sim i}^a)| < d_{\min}^a -1\), it follows that \( \mathrm{supp}(\boldsymbol{x}_{\sim i}^a) \not\subseteq \mathrm{Er}(\boldsymbol{m}_{\sim i}^a) \). Therefore, the message \( m_{a \rightarrow v_i^a} \) must be equal to 0.

Finally, we show that $P_{a \rightarrow v_i^a} = \Theta(\varepsilon^{d_{\min}^a-1})$. For any codeword \( \boldsymbol{x}^a \in \mathcal{C}_a \) with weight \( d_{\min}^a \) and satisfying \( x_i^a = 1 \), let  \( \boldsymbol{m}_{\sim i}^a \) satisfy \( \mathrm{Er}(\boldsymbol{m}_{\sim i}^a) = \mathrm{supp}(\boldsymbol{u}_{\sim i}^a) \). Then \( \boldsymbol{m}_{\sim i}^a \) satisfies Condition 1, and $P(\boldsymbol{m}_{\thicksim i}^a|\boldsymbol{0}^{n_a-1}) = \Theta\left(\varepsilon^{d_{\min}^a - 1}\right).$ It thus follows that $P_{a \rightarrow v_i^a} = \Theta(\varepsilon^{d_{\min}^a-1})$.

Consequently, there exists a constant $f(i,\mathcal{C}_a)$ such that $P_{a \rightarrow v_i^a}$ can be expressed as

\begin{equation}
 P_{a \rightarrow v_i^a} = g(i,\mathcal{C}_a)\, \varepsilon^{d_{\min}^a-1} + o(\varepsilon^{d_{\min}^a-1}).
 \label{e6}
\end{equation}
and $g(i,\mathcal{C}_a)>0$ under the GLDPC code ensemble.
\end{proof}

Next, we consider the impact of the update order between two constraint nodes that are adjacent in the scheduling sequence on the decoding performance. 

\begin{theorem}
\label{theo1}
For GLDPC codes that satisfy Assumption \ref{assump_3}, if $d_{\min}^a< d_{\min}^b$, then updating constraint node $b$ before $a$ results in a lower average erasure probability across the variable nodes, compared to updating $a$ before $b$.
\end{theorem}
\begin{proof}
Assuming that the all-zero codeword is transmitted. We first analyze the decoding performance when the constraint node $a$ is updated before $b$. For the constraint node $a$ associated with a subcode $\mathcal{C}_a$, according to Lemma \ref{lemma_biterror_allp}, the erasure probability of the message passed from $a$ to $v_i^a$ is given by:
\begin{equation}
P_{a \rightarrow v_i^a}= g(i,\mathcal{C}_a)\, \varepsilon^{d_{\min}^a-1} + o(\varepsilon^{d_{\min}^a}-1),
\end{equation}

Then, after updating constraint node $a$, let $\boldsymbol{m}^b$ denote the message transmitted from the neighboring variable nodes of $b$ to $b$. The erasure probability of messages received by $b$ from variable nodes $v_k^b$ with $k \in \{1, \ldots, n_{ab}\}$ is $\Theta(\varepsilon^{d_{\min}^a})$, while the erasure probability of messages from variable nodes $v_k^b$ with $k \in \{n_{ab}+1, \ldots, n_b\}$ is $\Theta(\varepsilon)$. Similar to Lemma \ref{lemma_biterror_allp}, we obtain that $
P(\boldsymbol{m}_{\sim j}^b \mid \boldsymbol{0}^{n_b-1}) = O\!\left(\varepsilon^{d_{\min}^b-1}\right),$ and it is $
\Theta\left(\varepsilon^{d_{\min}^b-1}\right)$ if and only if exactly $d_{\min}^b-1$ messages $m_{v_k^b \rightarrow b}=?$ for $k \neq j$, all of which satisfy $k \in \{n_{ab}+1,\ldots,n_b\}$, and $\boldsymbol{m}_{\sim j}^b$ satisfies Condition 1. Therefore, there exists a constant $h(j,n_{ab},\mathcal{C}_b)$ such that the erasure probability of the message passed from $b$ to $v_j^b$ is 
given by:

\begin{equation}
 P_{b \rightarrow v_j^b}= h(j,n_{ab},\mathcal{C}_b)\varepsilon^{d_{\min}^b-1}+o(\varepsilon^{d_{\min}^b-1}).
 \label{eq_1}
\end{equation}
and $h(j,n_{ab},\mathcal{C}_b) > 0$ under the GLDPC code ensemble assumption and the condition $n_b - n_{ab} > d_{\min}^b$, and $g(j,\mathcal{C}_b) \geq h(j,n_{ab},\mathcal{C}_b)$, with equality holding under the GLDPC code ensemble assumption if and only if $n_{ab} = 0$.

Therefore, after sequentially updating constraint nodes $a$ and $b$, the sum of the erasure probabilities over all variable nodes connected to $a$ or $b$ is given by:

\begin{equation}
\begin{split}
   P_{\text{sum}}^{ab}&=\sum_{i=n_{ab}+1}^{n_a}g(i,\mathcal{C}_a)\varepsilon^{d_{\min}^a}\\
   &+\sum_{j=n_{ab}+1}^{n_b}h(j,n_{ab},\mathcal{C}_b)\varepsilon^{d_{\min}^b}+o(\varepsilon^{d_{\min}^a}).
\end{split}
\label{e10}
\end{equation}

Similarly, the sum of the erasure probabilities after sequentially updating constraint nodes \(b\) and \(a\), denoted by \(P_{\text{sum}}^{ba}\), can be obtained. When \(d_{\min}^a < d_{\min}^b\), we have \(P_{\text{sum}}^{ab} \ge P_{\text{sum}}^{ba}\), which implies that updating \(b\) before \(a\) yields better decoding performance.
\end{proof}

Based on Theorem \ref{theo1}, it can be seen that subcodes with larger minimum distance should be updated first. When subcodes have the same minimum distance, further characterization of the distance properties—such as comparing the number of minimum-weight codewords—can be used to further confirm that subcodes with better distance properties should be prioritized for updating.

\begin{theorem}
\label{theo2}
In the first iteration, let $\varepsilon_i^a = \varepsilon$ for all $i$, and $\varepsilon_j^b = \varepsilon$ for all $j$ where $\varepsilon \to 0$. If any one of the following conditions holds, then updating constraint node $b$ before $a$ results in a lower average error probability across the variable nodes, compared to updating $a$ before $b$:

\noindent $\operatorname{1)} \; d_{\min}^a<d_{\min}^b$;

\noindent $\operatorname{2)}\;d_{\min}^a =d_{\min}^b$, and 
\begin{equation}
    \begin{split}
f(A_{\min}^a,\,d_{\min}^a,n_a,n_{ab})>f(A_{\min}^b,d_{\min}^b,n_b,n_{ab})   
\end{split}
\label{e2}
\end{equation}
where $A_{\min}^a$ denotes the number of codewords in $\mathcal{C}_a$ of weight $d_{\min}^a$, $f(A_{\min}^a,d_{\min}^a,n_a,n_{ab})=(n_a-n_{ab})\frac{A_{\min}^a}{{n_a \choose d_{\min}^a}} \bigg({n_a-1 \choose d_{\min}^a-1}-\\{n_a-n_{ab}-1 \choose d_{\min}^a-1}\bigg)$. When $n_a = n_b$, (\ref{e2}) reduces to $A_{{\min}}^a > A_{{\min}}^b$.
\end{theorem}
\begin{proof}[Proof of Theorem \ref{theo2}]
For the case $d_{\min}^a < d_{\min}^b$, the proof follows exactly the same procedure as that of Theorem \ref{theo1}. Therefore, we omit the repeated derivation and present the proof for the case $d_{\min}^a = d_{\min}^b$ below.

Assuming that the all-zero codeword is transmitted. We first analyze the decoding performance when the constraint node $a$ is updated before $b$. For the constraint node $a$ associated with a subcode $\mathcal{C}_a$, according to Lemma \ref{lemma_biterror_allp}, the erasure probability of the message passed from $a$ to $v_i^a$ is given by:
\begin{equation}
\begin{split}
 &P_{a \rightarrow v_i^a}=\sum_{\boldsymbol{m}_{\thicksim i}^a : m_{a \rightarrow v_i^a}=?}P(\boldsymbol{m}_{\thicksim i}^a|\boldsymbol{0}^{n_a-1})\\
 &=\sum_{\boldsymbol{m}_{\thicksim i}^a : m_{a \rightarrow v_i^a}=?,|\text{Er}(\boldsymbol{m}_{\sim i}^a)|\ge d_{\min}^a-1}P(\boldsymbol{m}_{\thicksim i}^a|\boldsymbol{0}^{n_a-1})
\end{split}
\end{equation}
where $|\text{Er}(\boldsymbol{m}_{\sim i}^a)|$ denotes the number of erased bits in the message \(\boldsymbol{m}_{\sim i}^a\). Moreover, based on (\ref{eq_1}), we have:

\begin{equation}
\begin{split}
    &\sum_{\boldsymbol{m}_{\thicksim i}^a : m_{a \rightarrow v_i^a}=?, |\text{Er}(\boldsymbol{m}_{\sim i}^a)|=d_{\min}^a-1}P(\boldsymbol{m}_{\thicksim i}^a|\boldsymbol{0}^{n_a-1}) \\
    &= \sum_{\boldsymbol{m}_{\thicksim i}^a : m_{a \rightarrow v_i^a}=?, |\text{Er}(\boldsymbol{m}_{\sim i}^a)|=d_{\min}^a-1}  \varepsilon^{d_{\min}^a-1} + o(\varepsilon^{d_{\min}^a-1}),
\end{split}
\end{equation}
and
\begin{equation}
\begin{split}
    &\sum_{\boldsymbol{m}_{\thicksim i}^a : m_{a \rightarrow v_i^a}=?, |\text{Er}(\boldsymbol{m}_{\sim i}^a)|>d_{\min}^a-1}P(\boldsymbol{m}_{\thicksim i}^a|\boldsymbol{0}^{n_a-1}) \\
    &= \sum_{\boldsymbol{m}_{\thicksim i}^a : m_{a \rightarrow v_i^a}=?, |\text{Er}(\boldsymbol{m}_{\sim i}^a)|>d_{\min}^a-1}  o(\varepsilon^{d_{\min}^a-1}).
\end{split}
\end{equation}
Therefore,
\begin{equation}
P_{a \rightarrow v_i^a}= g(i,\mathcal{C}_a)\, \varepsilon^{d_{\min}^a-1} + o(\varepsilon^{d_{\min}^a}-1),
\end{equation}
where \(g(i,\mathcal{C}_a)\) denotes the number of message vectors \(\boldsymbol{m}_{\sim i}^a\) such that \(m_{a \rightarrow v_i^a} = ?\) and \(|\operatorname{Er}(\boldsymbol{m}_{\sim i}^a)| = d_{\min}^a - 1\), i.e., the number of messages with exactly \(d_{\min}^a - 1\) erasures that satisfy Condition 1. 

This quantity is equal to the number of codewords $\boldsymbol{x}^a \in \mathcal{C}_{a}$ with Hamming weight $d_{\min}^a$ such that $x_i^a=1$. From an ensemble perspective, it can be expressed as
\begin{equation}
    g(i,\mathcal{C}_a) = \frac{A_{\min}^a}{{n_a \choose d_{\min}^a}}{n_a-1 \choose d_{\min}^a-1},
    \label{e_17}
\end{equation}
where $A_{\min}^a$ denotes the number of codewords in $\mathcal{C}_a$ of weight $d_{\min}^a$, the ratio $\frac{A_{\min}^a}{{n_a \choose d_{\min}^a}}$ represents the probability that a weight-$d_{\min}^a$ vector is a codeword. Furthermore, ${n_a-1 \choose d_{\min}^a-1}$ denotes the number of binary vectors of length $n_a$ and weight $d_{\min}^a$ whose $i$-th component is equal to 1.

After updating constraint node $a$, we proceed to update constraint node $b$. According to Lemma \ref{lemma_biterror_allp} and Theorem \ref{theo1}, the erasure probability of the message passed from $b$ to $v_j^b$ is given by:
\begin{equation}
\begin{split}
 &P_{b \rightarrow v_j^b}=\sum_{\boldsymbol{m}_{\thicksim j}^b : m_{b \rightarrow v_j^b}=?} \prod_{k \neq j} P_{v_k^b \rightarrow b}
\end{split}
\label{e_15}
\end{equation}
where
\begin{equation}
    P_{v_k^b \rightarrow b} = \begin{cases}
       P_{a \rightarrow v_k^b}\varepsilon, &k\in\{1,2,\ldots,n_{ab}\} \text{, and } m_{v_k \rightarrow b}=? \\
       1-P_{a \rightarrow v_k^b}\varepsilon, &k\in\{1,2,\ldots,n_{ab}\} \text{, and } m_{v_k \rightarrow b}=0 \\
       \varepsilon,&k\in\{n_{ab}+1,\ldots,n_{b}\}\text{, and } m_{v_k \rightarrow b}=?\\
       1-\varepsilon,&k\in\{n_{ab}+1,\ldots,n_{b}\}\text{, and } m_{v_k \rightarrow b}=0
    \end{cases}
    \label{e_16}
\end{equation}


From \eqref{e_15}–\eqref{e_16}, it holds that
\[\prod_{k \neq j} P_{v_k^b \rightarrow b}=\varepsilon^{d_{\min}^b-1}+o(\varepsilon^{d_{\min}^b-1})\]
if and only if exactly $d_{\min}^b-1$ messages $m_{v_k \rightarrow b}=?$ for $k \neq j$, all of which satisfy $k \in \{n_{ab}+1,\ldots,n_b\}$, and $\boldsymbol{m}_{\sim j}^b$ satisfies Condition 1. Otherwise,
\[
\prod_{k \neq j} P_{v_k^b \rightarrow b} = o\left(\varepsilon^{d_{\min}^b-1}\right).
\]
Therefore,
\begin{equation}
P_{b \rightarrow v_j^b}= h(j,n_{ab},\mathcal{C}_b)\, \varepsilon^{d_{\min}^b-1} + o(\varepsilon^{d_{\min}^b}-1),
\end{equation}
where \(h(j,n_{ab},\mathcal{C}_b)\) denotes the number of message vectors \(\boldsymbol{m}_{\sim j}^b\) that \(\boldsymbol{m}_{\sim j}^b\) satisfy Condition 1, have exactly \(|\operatorname{Er}(\boldsymbol{m}_{\sim j}^b)| = d_{\min}^b - 1\) erasures, and whose erasure set satisfies
\[\operatorname{Er}(\boldsymbol{m}^b) \subseteq \{n_{ab}+1 ,\ldots,n_b\} \cup \{j\}.\]

This quantity is equal to the number of codewords $\boldsymbol{x}^b \in \mathcal{C}_{b}$ with Hamming weight $d_{\min}^b$ such that $x_j^b=1$ and
\[\text{supp}(\boldsymbol{x}^b) \subseteq \{n_{ab}+1 ,\ldots,n_b\} \cup \{j\}.\]
From an ensemble perspective, it can be expressed as
\begin{equation}
    h(j,n_{ab},\mathcal{C}_b)=\begin{cases}
\frac{A_{\min}^b}{{n_b \choose d_{\min}^b}}{n_b-n_{ab} \choose d_{\min}^b-1}, & \text{if } j \in \{1,\ldots,n_{ab}\}, \\
\frac{A_{\min}^b}{{n_b \choose d_{\min}^b}}{n_b-n_{ab}-1 \choose d_{\min}^b-1}, &\text{ otherwise.}
\end{cases}
\label{e_17}
\end{equation}
where $A_{\min}^b$ denotes the number of codewords in $\mathcal{C}_b$ of weight $d_{\min}^b$. The ratio $\frac{A_{\min}^b}{{n_b \choose d_{\min}^b}}$ represents the probability that a weight-$d_{\min}^b$ vector is a codeword. Furthermore, the number of binary vectors $\boldsymbol{u}$ of length $n_b$ and weight $d_{\min}^b$ satisfying \( u_j^b = 1 \) and \( \operatorname{supp}(\boldsymbol{u}^b) \subseteq \{n_{ab}+1, \ldots, n_b\} \cup \{j\} \) is given by
\( {n_b-n_{ab} \choose d_{\min}^b-1} \quad \text{for } j \in \{1, \ldots, n_{ab}\},\) and \({n_b-n_{ab}-1 \choose d_{\min}^b-1}\) otherwise.

Therefore, we obtain the sum of the erasure probabilities over all variable nodes connected to \(a\) or \(b\) after sequentially updating constraint nodes \(a\) and \(b\)

\begin{equation}
\begin{split}
&P_{\text{sum}}^{ab}=\sum_{i=n_{ab}+1}^{n_a}g(i,\mathcal{C}_a)\varepsilon^{d_{\min}^a}+\sum_{j=n_{ab}+1}^{n_b}h(j,n_{ab},\mathcal{C}_b)\varepsilon^{d_{\min}^b}\\
&+o(\varepsilon^{d_{\min}^a})\\
 &=\bigg((n_a-n_{ab})\frac{A_{\min}^a}{{n_a \choose d_{\min}^a}}{n_a-1 \choose d_{\min}^a-1}\\
 & +(n_b-n_{ab}) \frac{A_{\min}^b}{{n_b \choose d_{\min}^b}}{n_b-n_{ab}-1 \choose d_{\min}^b-1}\bigg)\varepsilon^{d_{\min}^a}+o(\varepsilon^{d_{\min}^a}). 
\end{split}
\end{equation}

Similarly, \(P_{\text{sum}}^{ba}\) can be derived for the reverse updating order 
\begin{equation}
\begin{split}
&P_{\text{sum}}^{ab}=\bigg((n_b-n_{ab})\frac{A_{\min}^b}{{n_b \choose d_{\min}^b}}{n_b-1 \choose d_{\min}^b-1}\\
&+(n_a-n_{ab}) \frac{A_{\min}^a}{{n_a \choose d_{\min}^a}}{n_a-n_{ab}-1 \choose d_{\min}^a-1}\bigg)\varepsilon^{d_{\min}^a}+o(\varepsilon^{d_{\min}^a}). 
\end{split}
\end{equation}

When the condition $f(A_{\min}^a,d_{\min}^a,n_a,n_{ab})>f(A_{\min}^b,d_{\min}^b,n_b,n_{ab})$ holds, it follows that \(
P_{\text{sum}}^{ab} \ge P_{\text{sum}}^{ba}.\)
\end{proof}

The proof of Theorem \ref{theo2} follows similar arguments to that of Theorem \ref{theo1} and is omitted due to space limitations. These two theorems show that, for any two adjacent subcodes in the decoding order, the constraint node associated with the subcode having better distance properties should be updated first. For constraint nodes that are not adjacent in the sequence, the same objective can be achieved by successively exchanging the update order of multiple adjacent constraint nodes.

\begin{remark}
Although Theorems \ref{theo1}–\ref{theo2} are established under the assumption of the first iteration, simulation results suggest that similar qualitative conclusions also hold for subsequent iterations.
\end{remark}

\section{Simulation}

\begin{figure*}[htbp]
	\centering
	\begin{subfigure}{0.325\linewidth}
		\centering
		\includegraphics[width=1.1\linewidth]{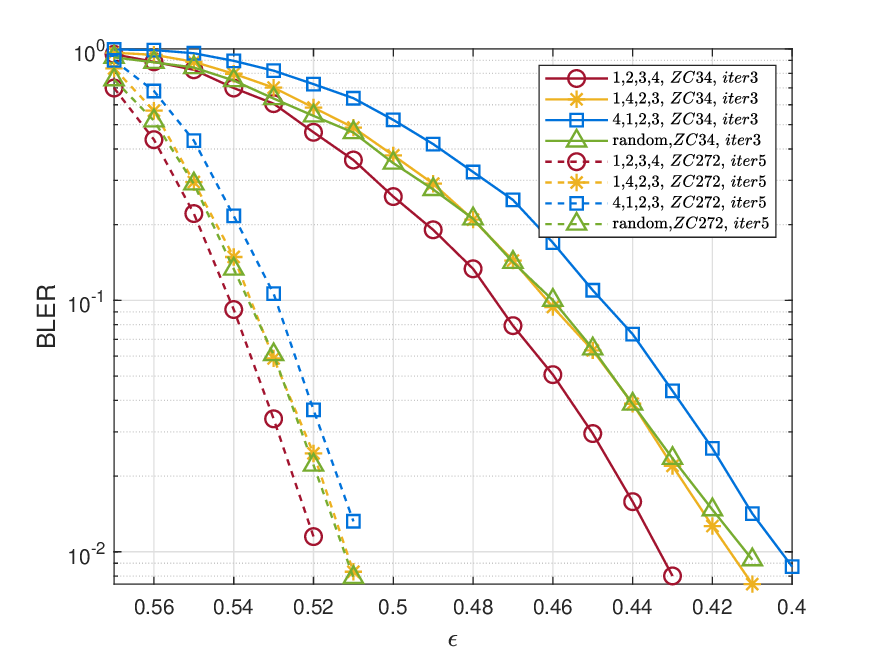}
		\caption{Simulation on the GLDPC code with (7,4,3) Hamming subcode over the BEC channel. \quad\quad\quad\quad}
		\label{fig_5}
	\end{subfigure}
	\centering
    \begin{subfigure}{0.325\linewidth}
		\centering
		\includegraphics[width=1.1\linewidth]{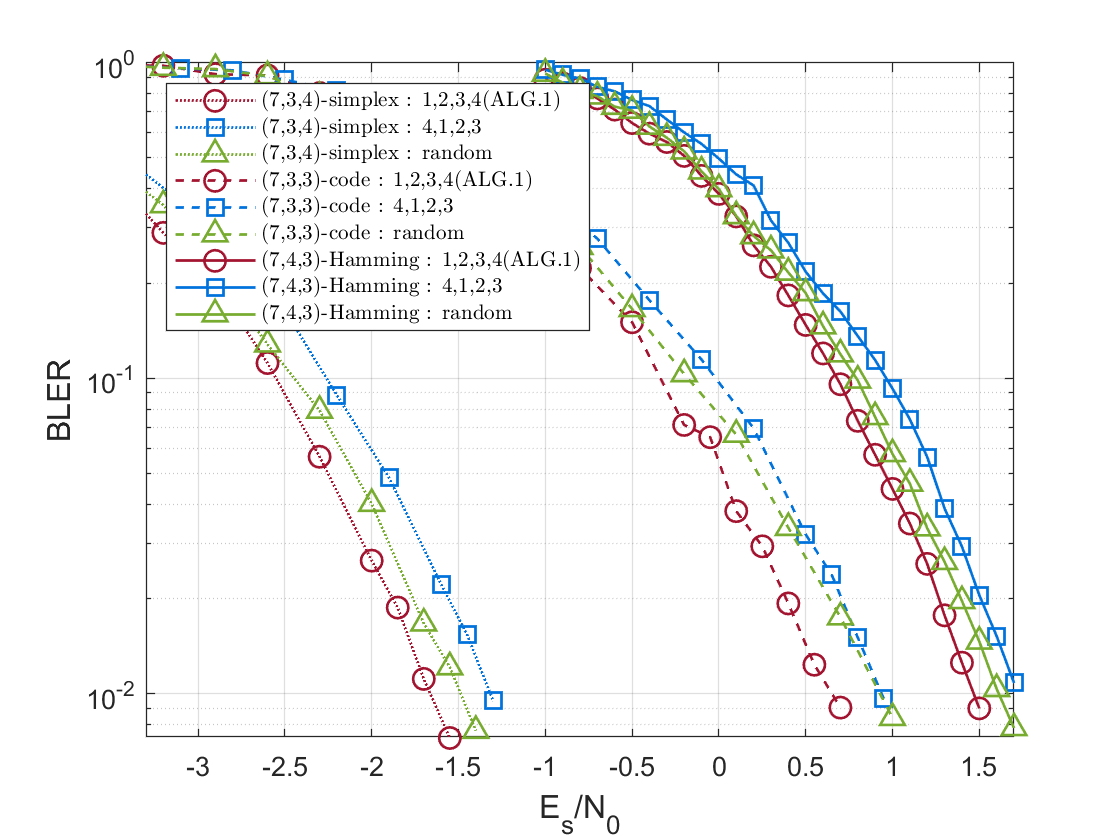}
		\caption{Simulation on GLDPC codes with lifting size ZC=34 and 3 iterations over the BI-AWGN channel.}
		\label{fig_2}
	\end{subfigure}
 \\
 
    \begin{subfigure}{0.325\linewidth}
		\centering
\includegraphics[width=1.1\linewidth]{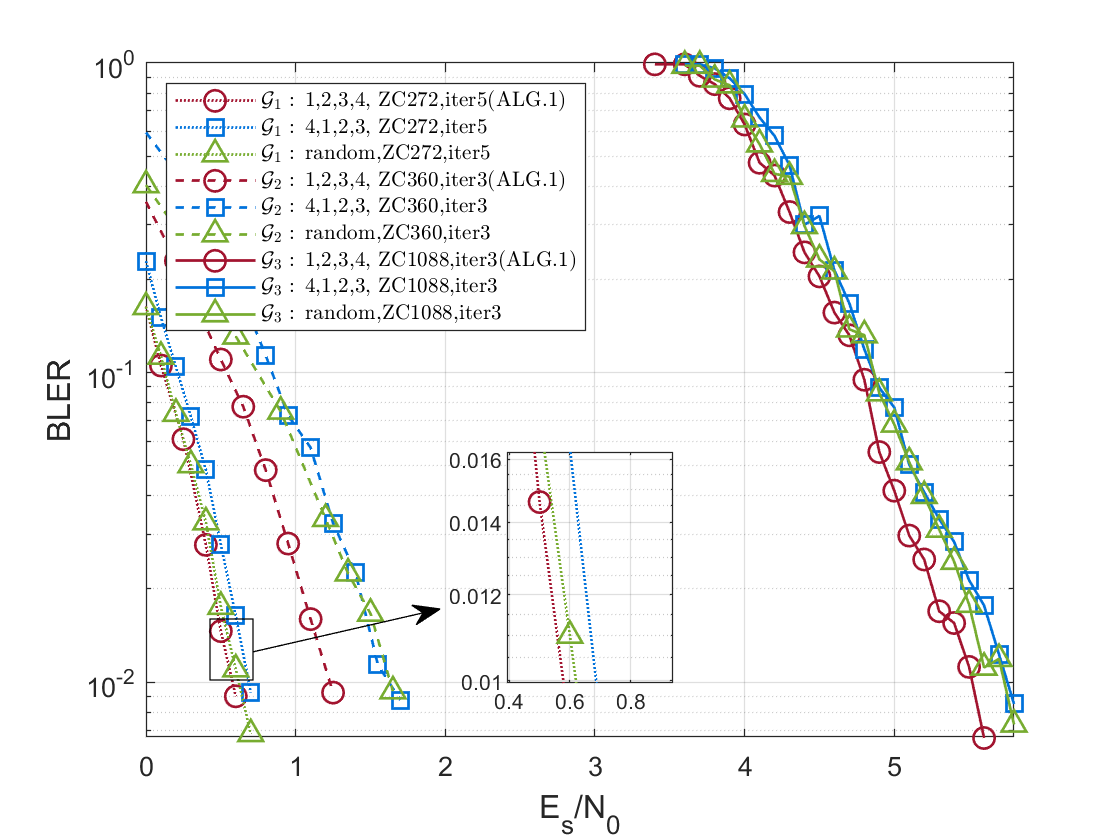}
		\caption{Simulations on GLDPC codes \(\mathcal{G}_1\), \(\mathcal{G}_2\), and \(\mathcal{G}_3\) over the BI-AWGN channel.}
		\label{fig_2}
	\end{subfigure}
\centering
    \begin{subfigure}{0.325\linewidth}
		\centering
\includegraphics[width=1.1\linewidth]{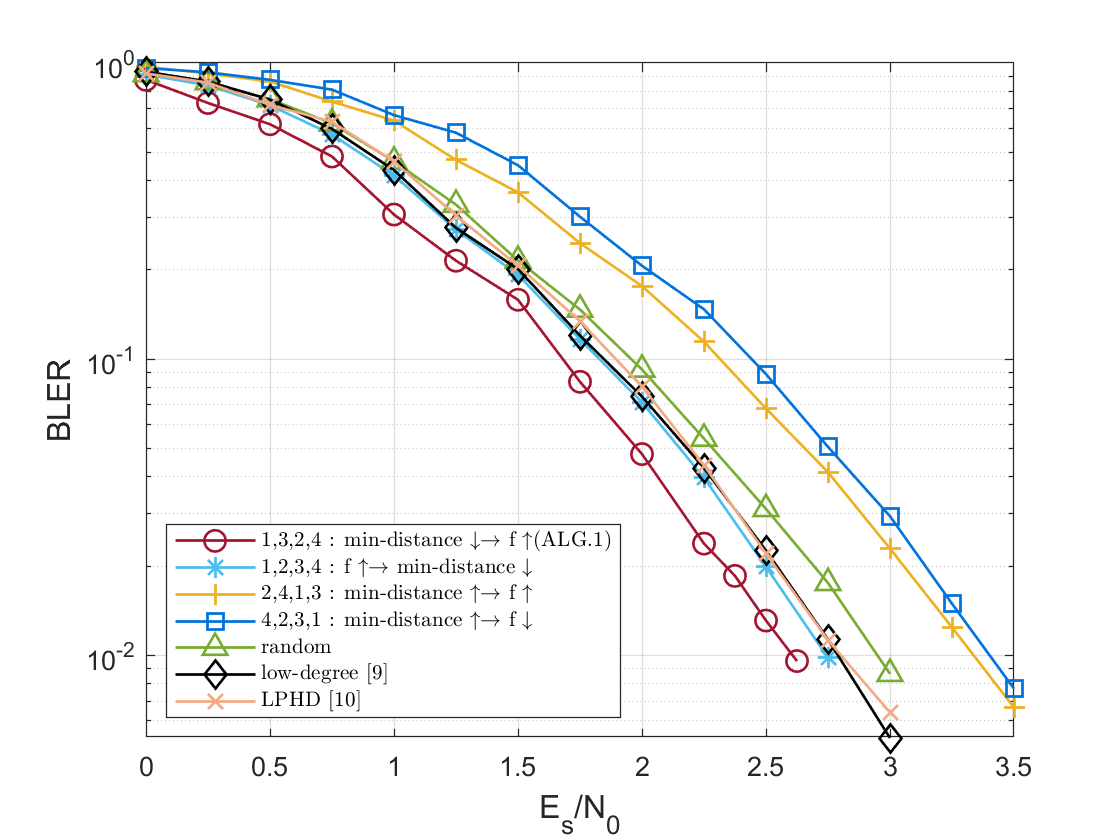}
		\caption{Simulation on GLDPC codes with lifting size ZC=45 and 3 iterations}
		\label{fig_2}
	\end{subfigure}
	\caption{BLER of the proposed scheduling algorithm and of layered decoding for the GLDPC codes over the BI-AWGN channel}
	\label{fig}
\end{figure*}

In this section, we design an algorithm to identify scheduling sequences that yield improved decoding performance based on Theorem 2 and Theorem 5.  As summarized in Algorithm 1, the constraint nodes are updated in descending order of the minimum distance of their associated subcodes. For nodes with equal minimum distance, the ordering follows ascending values of the function $f$, where constraint nodes with fewer minimum-weight codewords and shorter code length correspond to smaller values of $f$. Note that Algorithm 1 is an offline algorithm executed in advance, and the scheduling sequence $\beta$ is precomputed and stored without increasing the complexity of encoding and decoding. 

\begin{algorithm}
\caption{Hierarchical Distance Scheduling}
\begin{algorithmic}[1]
\Require the parity-check matrix and the subcodes associated with each constraint node
\Ensure scheduling sequence $\beta$

\State Compute the sequences \(d\) and \(A\) for all constraint nodes, representing the minimum distance and minimum-weight codeword count, respectively.
\State $\beta \leftarrow 1$;
\For{$i=2:$ length($d$)}
{
\State $\beta \leftarrow [\beta,i]$;
\For {$j =$length$(\beta):-1 :2$} 
    \If{$d_{\beta_j}>d_{\beta_{j-1}}$}
    \State swap $\beta_j$ and $\beta_{j-1};$
    \ElsIf{$f(A_{\beta_{j-1}},d_{\beta_{j-1}},n_{\beta_{j-1}},n_{\beta_{j}\beta_{j-1}})>f(A_{\beta_j},d_{\beta_j},n_{\beta_j},n_{\beta_j\beta_{j-1}})$ and $d_{\beta_j}=d_{\beta_{j-1}}$ } 
    \State  swap $\beta_j$ and $\beta_{j-1};$
    \Else
    \State break;
    \EndIf
\EndFor
}
\EndFor
\end{algorithmic}
\end{algorithm}

For quasi-cyclic low-density parity-check codes, let $P$ be the exponent matrix with entries \( P_{i,j} \in \{-1, 0, \ldots, ZC-1\} \), where \( ZC \) is the lifting size. The corresponding parity-check matrix is formed by replacing each entry \( P_{i,j}\) with a \( ZC \times ZC \) circulant permutation matrix with shift \( P_{i,j} \) if \( P_{i,j} \ge 0 \), and with the all-zero matrix otherwise. To describe the decoding schedule, we represent the scheduling sequence using the row indices of the exponent matrix. In this representation, each index $\alpha$ indicating that the lifted check code corresponding to the $\alpha$th row of the exponent matrix is updated at that position in the sequence.

In Fig. \ref{fig}, we present the simulation performance of the scheduling sequences designed by the proposed algorithm. The simulations are based on the GLDPC codes constructed in \cite{10619445}, designed to achieve relatively large minimum distances. Let the exponent matrices of the LDPC codes in Section III-B, Example 3 of Section III-A, and Section III-C in \cite{10619445} be denoted by \(G_{R_4}^{(1)}\), \(G_{R_4}^{(2)}\), and \(G_{R_4}^{(3)}\), respectively. The GLDPC codes used in the simulations are obtained by generalizing the constraint nodes based on these exponent matrices.

\begin{table}[htbp]
\caption{the exponent matrix of {$G_{R_4}^{(1)}$}}
\centering
\scalebox{0.9}{
 \begin{tabular}{|cc|cc|cc|cc|cc|cc|cc|}
\hline
0 & -1 & 0 & -1 & 0 & -1 & 0 & -1 & 0 & -1 & 0 & -1&0&-1 \\
-1 & 0 & -1 & 0 & -1 & 0 & -1 & 0 & -1 & 0&-1&0&-1&0 \\
\hline
0 & -1&-1 & 31 & -1 & 25 & 22 & -1&-1 & 1 & 23 &-1& 7&-1 \\
-1 & 0 & 30&-1 & 24 & -1 & -1 & 22 & 0 & -1&-1 & 23&-1&7 \\
\hline
\end{tabular}
}
\label{tab1}
\end{table}

\begin{table}[htbp]
\caption{the exponent matrix of {$G_{R_4}^{(2)}$}}
  \centering
  \scalebox{0.9}{
\begin{tabular}{|cc|cc|cc|cc|cc|cc|}
\hline
0 & -1 & 0 & -1 & 0 & -1 & 0 & -1 & 0 & -1 & 0 & -1 \\
-1 & 0 & -1 & 0 & -1 & 0 & -1 & 0 & -1 & 0&-1&0 \\
\hline
0 & -1&27 & -1 & 33 & -1 & -1& 36 & -1 & 28&-1& 35 \\
-1 & 0 &-1& 27&-1 & 33 & 35 & -1 & 27 & -1 & 34&-1 \\
\hline
\end{tabular}
  }
  \label{tab2}
\end{table}

\begin{table*}[htbp]
\caption{the exponent matrix of {$G_{R_4}^{(3)}$}}
  \centering
  \scalebox{0.9}{
\begin{tabular}{|cc|cc|cc|cc|cc|cc|cc|cc|cc|cc|cc|cc|cc|cc|cc|}
\hline
0 & -1 & 0 & -1 & 0 & -1 & 0 & -1 & 0 & -1 & 0 & -1&0 & -1 & 0 & -1 & 0 & -1 & 0 & -1 & 0 & -1 & 0 & -1 & 0 & -1& 0 & -1& 0 & -1\\
-1 & 0 & -1 & 0 & -1 & 0 & -1 & 0 & -1 & 0 & -1&0 & -1 & 0 & -1 & 0 & -1 & 0 & -1 & 0 & -1 & 0 & -1 & 0 & -1& 0 & -1& 0 & -1 &0\\
\hline
0 & -1&3&-1&9&-1&-1&12&-1&4&-1&11&-1&1&-1&2&-1&5&-1&14&-1&10&14&-1&8&-1&-1&9&10&-1\\
-1 & 0 &-1& 3&-1&9&11&-1&3&-1&10&-1&0&-1&1&-1&4&-1&13&-1&9&-1&-1&14&-1&8&8&-1&-1&10\\
\hline
\end{tabular}
  }
  \label{tab2}
\end{table*}

\begin{table}[htbp]
\caption{the exponent matrix of {$G_{R_4}^{(4)}$}}
  \centering
  \scalebox{0.9}{
\begin{tabular}{|cc|cc|cc|cc|cc|cc|}
\hline
0 & -1 & 0 & -1 & 0 & -1 & 0 & -1 & 0 & -1 & 0 & -1 \\
-1 & 0 & -1 & 0 & -1 & 0 & -1 & 0 & -1 & 0&-1&0 \\
\hline
0 & 4&27 & -1 & 33 & -1 & -1& 36 & -1 & 28&-1& 35 \\
12 & 0 &-1& 27&-1 & 33 & 35 & -1 & 27 & -1 & 34&-1 \\
\hline
\end{tabular}
  }
  \label{tab4}
\end{table}



In Fig. 1(a), we present the simulation results of the GLDPC codes obtained by generalizing the first three constraint nodes in the exponent matrix $G_{R_4}^{(1)}$ using (7,4,3) Hamming code over the BEC. The subcodes in the first three rows have a minimum distance of 3, while the subcode in the fourth row has a minimum distance of 2. For comparison, three illustrative scheduling sequences are selected based on the position of the fourth row in the decoding order: \((1,2,3,4)\), \((1,4,2,3)\), and \((4,1,2,3)\), along with a random sequence, where 'random' indicates that, for each decoding instance, one sequence is chosen uniformly at random from the 24 possible sequences. Simulations are conducted on GLDPC with lifting sizes $ZC$ of 34 and 272, corresponding to code rates of \(2/7\) and code lengths of 476 and \(3808\), respectively. The maximum number of decoding iterations is limited to three or five to evaluate performance under high-throughput scenarios. It can be observed that the scheduling sequence $(1,2,3,4)$, which prioritizes updating subcodes with better minimum-distance properties, demonstrates superior performance. The simulation results indicate that prioritizing the update of constraint nodes associated with subcodes of better distance properties yields superior decoding performance.

Although our analysis and simulations were carried out over the BEC, we expect similar conclusions to hold for more complex channels, such as the BI-AWGN channel. To verify this, we also conduct simulations over the BI-AWGN channel. In Fig. 1(b), we present the simulation results of the GLDPC codes obtained by generalizing the first three constraint nodes in the exponent matrix $G_{R_4}^{(1)}$ using different subcodes, including the (7,4,3) Hamming code, the (7,3,4) Simplex code (i.e., the dual of the Hamming code), and a (7,3,3) subcode of the Hamming code. The subcodes in the first three rows have a minimum distance of no less than 3, while the subcode in the fourth row has a minimum distance of 2. For comparison, two illustrative scheduling sequences are selected based on the position of the fourth row in the decoding order: \((1,2,3,4)\) and \((4,1,2,3)\), along with a random sequence. It can be observed that the scheduling sequence $(1,2,3,4)$, which prioritizes updating subcodes with better minimum-distance properties, demonstrates superior performance. The simulation results indicate that prioritizing the update of constraint nodes associated with subcodes of better distance properties yields superior decoding performance.

In Fig. 1(c), we present the performance of the GLDPC codes $\mathcal{G}_1$, $\mathcal{G}_2$, and $\mathcal{G}_3$, constructed by generalizing the first three rows of $G_{R_4}^{(1)}$, $G_{R_4}^{(2)}$, and $G_{R_4}^{(3)}$, respectively. The subcodes are the (7,4,3) Hamming code, the (6,3,3) Hamming code, and the (15,11,3) shortened Hamming code whose parity-check matrix is obtained by removing the all-one column from that of the (7,4,3) Hamming code. The scheduling sequence \((1,2,3,4)\), which prioritizes the constraint nodes with larger minimum distance, consistently achieves performance gains across different GLDPC codes, iteration numbers, and code lengths.

In Fig.~1(d), we construct an LDPC code \({G}_{R_4}^{(4)}\) and evaluate the performance of the corresponding GLDPC code \(\mathcal{G}_4\). The exponent matrix \({G}_{R_4}^{(4)}\) is obtained by adding two edges to \({G}_{R_4}^{(2)}\), where the shift values are chosen to minimize the shortest cycle in the Tanner graph. The GLDPC code \(\mathcal{G}_4\) is then constructed by generalizing the first row of \({G}_{R_4}^{(4)}\) into a \((6,3,3)\) shortened Hamming code and the third row into a \((7,4,3)\) Hamming code. Algorithm 1 outputs the scheduling sequence \((1,3,2,4)\). For comparison, the sequences \((1,2,3,4)\), \((4,2,3,1)\), and \((4,3,2,1)\), as well as a random sequence, and the low-degree \cite{frenzel2019static} and LPHD \cite{wang2020two} prioritized updating strategies, are also considered. The results show that \((1,3,2,4)\) achieves the best performance, followed by \((1,2,3,4)\). These two sequences are both constructed based on a lexicographic ordering of minimum distance (in descending order) and the
number of low-weight codewords and code length (in ascending order), but with different priority rules. In contrast, \((2,4,1,3)\) and \((4,2,3,1)\), which follow ascending minimum distance, exhibit the worst performance. This confirms that the minimum distance is the dominant factor, while the multiplicity of low-weight codewords and the code length serve as secondary factors.


\bibliographystyle{IEEEtran}
\bibliography{IEEEabrv,ref}

\end{document}